\documentclass[11pt]{article}
\usepackage{amsmath}
\usepackage{amssymb}
\usepackage{amsthm}
\usepackage{braket}
\usepackage[backend=biber,style=alphabetic,maxbibnames=99]{biblatex}
\usepackage{hyperref}
\usepackage{mathtools}
\usepackage{mathrsfs}

\usepackage{appendix}

\usepackage[hmargin=1in,bmargin=1.4in,tmargin=1in]{geometry}

\usepackage{graphicx}
\newcommand{\bigboxtimes}{\mathop{\raisebox{-0.35em}{\scalebox{1.75}{$\boxtimes$}}}}

\theoremstyle{plain}
\newtheorem{theorem}{Theorem}
\newtheorem{lemma}[theorem]{Lemma}
\newtheorem{proposition}[theorem]{Proposition}

\theoremstyle{definition}
\newtheorem{definition}[theorem]{Definition}
\newtheorem{warning}[theorem]{Warning}

\theoremstyle{remark}

\newtheorem{remark}[theorem]{Remark}
\newtheorem{example}[theorem]{Example}

\allowdisplaybreaks

\title{
A convergent hierarchy of spectral gap certificates for qubit Hamiltonians}
\author{
Sujit Rao
}
\date{
}
\begin{document}

\maketitle

\begin{abstract}
We give a convergent hierarchy of SDP certificates for lower bounds on the spectral gap of a local qubit Hamiltonian $H$. Our approach is based on the NPA hierarchy applied to a polynomially-sized system of constraints defining the universal enveloping algebra of the Lie algebra $\mathfrak{su}(2^{n})$ along with additional constraints which force the representation to be $\wedge^{2}(\mathbb{C}^{2^{n}})$, one which the induced Hamiltonian has ground energy $\lambda_{0}(H) + \lambda_{1}(H)$. Combined with an upper bound on the ground state energy, which can be obtained either using a hierarchy introduced by Fawzi, Fawzi, and Scalet or a noncommutative analog developed by Klep, Magron, Massé, and Volčič of a hierarchy introduced by Lasserre, we obtain lower bounds on the spectral gap of $H$ without assuming frustration-freeness or geometric locality.

We prove that the resulting certificates have polynomial size at fixed degree and converge at level $n$. As illustrations of the certificates which can be obtained, we show that for a commuting 1-local Hamiltonian the hierarchy certifies the optimal bound on its spectral gap and prove that similar hierarchies require degree growing with $n$ to certify the same bound. We then extend this to prove a lower bound on the spectral gap of the transverse-field $XY$ model, which is not frustration-free, on arbitrary regular graphs without any assumption of geometric locality.
\end{abstract}

\section{Introduction and motivation}
Certifying that a quantum Hamiltonian has a nonzero spectral gap is a fundamental problem in quantum many-body physics. In general, a strictly positive gap controls correlation decay and stability of phases. For 1D and 2D systems in particular, the existence of a gap can imply strong properties about entanglement in the system, such as the existence of an area law. Such properties are relevant for the design of algorithms for many-body systems, which in some cases use limitations on entanglement in the ground state to output good classical approximations.

However, proving rigorous lower bounds on the spectral gap is often difficult. From a computational complexity viewpoint, Gharibian and Yirka proved that computing the spectral gap of a local qubit Hamiltonian is hard for the complexity class $\mathsf{P}^{\mathsf{UQMA}[\log]}$ \cite{gy2019}. In certain models of infinite-size systems, computing the spectral gap is known to be undecidable \cite{cpw2015, bclp2020}. From a condensed matter physics viewpoint, several well-known open problems, such as the Haldane conjecture \cite{h1983, h1983a, h2016}, amount to showing that particular quantum systems are either gapped or gapless.

The goal of this paper is to present a \emph{convergent} hierarchy of semidefinite programming (SDP) certificates that can be used to certify lower bounds on spectral gaps of qubit Hamiltonians which may be frustrated or geometrically nonlocal. The hierarchy is constructed from a noncommutative polynomial system which gives a lower bound $A$ on the sum $\lambda_{0}(H) + \lambda_{1}(H)$ of the first two eigenvalues of a given local Hamiltonian. We also use a family of upper-bound SDPs to get an upper bound $B \geq \lambda_{0}(H)$ on the lowest eigenvalue of the Hamiltonian from above. We then obtain a lower bound of
\[ A - 2B \leq \lambda_{1}(H) + \lambda_{0}(H) - 2\lambda_{0}(H) = \lambda_{1}(H) - \lambda_{0}(H) \]
on the spectral gap by combining the bounds produced by the two hierarchies.

\subsection{Previous work}
Recent work on spectral gap bounds for spin systems using SDPs includes the spectral bootstrap method of \cite{nx2023}, which uses an approach inspired by the bootstrap method in conformal field theory to give \textit{upper} bounds on the spectral gaps of translationally-invariant qubit systems. The approach allows one to more generally deduce constraints on the matrix elements of local operators expanded in the eigenbasis of the system.

A recent SDP hierarchy for producing \emph{lower} bounds on gaps of frustration-free Hamiltonians was proposed in \cite{rketsb2024}. If the exact ground energy $\lambda_{0}(H)$ of the Hamiltonian is known, then a gap of at least $\gamma$ is equivalent to the operator $(H - \lambda_{0}(H))(H - \gamma - \lambda_{0}(H))$ being positive. One can then try to certify a lower bound on the spectral gap by searching over low-degree sum-of-squares certificates that this operator is positive. If $\lambda_{0}(H)$ is known, then the left-hand side becomes linear in $\gamma$ and thus amenable to SDP approaches based on the noncommutative sum-of-squares hierarchy. The same work also proposed several approximations and simplifications for translationally-invariant qubit systems on a lattice.

\subsection{Results}
\subsubsection{Our hierarchy}
We propose a hierarchy of SDP-based certificates for proving lower bounds on the spectral gaps of quantum qubit systems on a finite lattice. We first define a noncommutative polynomial optimization problem whose optimal value is equal to the sum of the first two eigenvalues of a given local qubit Hamiltonian, which can then be approximated using the NPA hierarchy \cite{npa2010a, npa2010b}. We then use a variant of the Lasserre hierarchy of upper bounds on polynomial minimization problems, introduced by Klep, Magron, Mass\'{e}, and Vol\v{c}i\v{c} in \cite{kmmv2025}, to bound the ground state energy of a local Hamiltonian from above. Taking an appropriate linear combination then gives a certified lower bound on the spectral gap of the Hamiltonian.

We prove both asymptotic convergence of the hierarchy and note that standard arguments imply that convergence in fact occurs at level $n$. The proof that the noncommutative polynomial system encodes the desired problem is quite involved, and uses a number of facts about the Lie algebra $\mathfrak{su}(2^{n})$ and its representations.

\subsubsection{New lower bounds}
To illustrate the bounds which this hierarchy can certify, we first consider the 1-local Hamiltonian
\[ H_{n} = -\sum _{i = 1} ^{n} Z_{i} \]
for which the spectral gap can be calculated directly. We show that our hierarchy certifies the optimal lower bound on $\lambda_{0}(H_{n}) + \lambda_{1}(H_{n})$ at level 2, and when combined with the hierarchy of Fawzi, Fawzi, and Scalet produced the optimal lower bound on the spectral gap. We also show that other SDP hierarchies, including those based on considering $(H - \lambda_{0}(H))(H - \gamma - \lambda_{0}(H))$ require degree growing with $n$ to certify a nontrivial bound.

We then extend our lower bound to the transverse-field $XY$ model with Hamiltonian
\[ H_{G}(h, J) = -h\sum _{i \in V(G)} Z_{i} - \frac{J}{2}\sum _{(i, j) \in E(G)} (X_{i}X_{j} + Y_{i}Y_{j}) \]
where $G$ is an arbitrary $d$-regular undirected graph, which can be seen as a perturbation of the 1-local Hamiltonian when $J/h$ is small. We use our hierarchy to prove a lower bound of $2h - dJ$ on the spectral gap when it is nonnegative, recovering known results based on a mapping to free fermions when $G$ is a one-dimensional cycle and extending them to higher-dimensional lattices.

\subsection{Further directions}
The size of the noncommutative polynomial system we construct is still rather large, with roughly $O(n^{2})$ variables and $O(n^{8})$ constraints and nontrivial constant factors. One main question is whether a simpler polynomial system can still give asymptotic convergence to the true spectral gap. It seems unclear the same approach of defining generators and relations for $\mathfrak{su}(2^{n})$ as a Lie algebra can work, since the proof strategy in this paper seems to rely on all of the relations included in the system. But it is still possible that some other set of constraints could suffice to give complete certificates of lower bounds.

Our hierarchy also suggests an approach to compute bounds on the first $k$ eigenvalues of $H$ for constant values of $k$. This would require finding constraints which characterize the higher exterior powers $\wedge^{k}(\mathbb{C}^{2^{n}})$ as representations of $\mathfrak{su}(2^{n})$ beyond $k = 2$, although even for $k = 2$ the constraints we give here are already rather complicated.

\section{Technical overview}

\subsection{Hierarchy for lower bounds}
The lower-bound SDP hierarchy is built from a noncommutative $*$-algebra on a finite set of generators (with variables denoted $S_i^a$ and $T_{ij}^{ab}$), together with the explicit relations recorded in \eqref{eq:rels} and also the spectral and representation-theoretic constraints in \eqref{eq:sys}. Intuitively, $S_{i}^{a}$ is meant to correspond to a 1-local Pauli operator, and $T_{ij}^{ab}$ corresponds to a 2-local Pauli operator.

Concretely, we first define the algebra $\mathbb{C}\langle V_n\rangle / \mathcal I$ where $\mathcal I$ is the two-sided $*$-ideal generated by the relations \eqref{eq:rels}. These relations enforce certain commutation relations which characterize the universal enveloping algebra of the Lie algebra $\mathfrak{su}(2^{n})$. Its representations then correspond to irreducible Lie algebra representations of $\mathfrak{su}(2^{n})$, of which one is the antisymmetric subspace $\wedge^{2}(\mathbb{C}^{2^{n}})$. The action of $H$ on the antisymmetric subspace has $\lambda_{0}(H) + \lambda_{1}(H)$ as its lowest eigenvalue, though there are other representations. The definitions of the generators and relations and proof that these characterize $\mathfrak{su}(2^{n})$ are in Section~\ref{sec:genrel}.

We then impose additional noncommutative polynomial constraints and inequalities to rule out representations other than the antisymmetric subspace. For example, in \eqref{eq:sys} there are constraints on the eigenvalues of $S_{i}^{a}$, which then force the irreducible representation of $\mathfrak{su}(2^{n})$ to correspond to one of five particular Young diagrams. Additional constraints then rule out each of the four remaining irreducible representations other than $\wedge^{2}(\mathbb{C}^{2^{n}})$. The exact definitions of these additional constraints and the proof that they rule out all other irreps, appear in Section~\ref{sec:hierarchy}.

One slightly subtle point is that the additional two-site variables $T_{ij}^{ab}$ are needed to encode higher-locality interactions. Specifically, any $k$-local term in a Hamiltonian can be expressed as a Lie polynomial (which can be expanded to a noncommutative polynomial) in the chosen generators. The action of $H$ on a representation of the algebra we define is then that of the Lie algebra representation, which is the intended operator for $\wedge^{2}(\mathbb{C}^{2^{n}})$. However, writing the Hamiltonian as a noncommutative polynomial in the $S_{i}^{a}$ would not give a correctly converging hierarchy of lower bounds on $\lambda_{0}(H) + \lambda_{1}(H)$.

\subsection{Hierarchy of upper bounds}
One hierarchy of upper bounds uses the noncommutative version of the Lasserre hierarchy introduced in \cite{kmmv2025}. Concretely, one optimizes over positive noncommutative polynomials $P$ of bounded degree and forms the reweighted state $\psi_P(\cdot)=\frac{\varphi(P\,\cdot)}{\varphi(P)}$ where $\varphi$ is a fixed reference state. Minimizing the expectation $\psi_P(H)$ over such $P$ yields a sequence of upper bounds $B_d$ on the ground energy which can be computed using an SDP. The resulting upper bounds decrease with $d$ and converge to the true ground energy. We give some more detail on this hierarchy in Section~\ref{sec:upperbounds-lasserre}.

Another is based on the KMS conditions for thermal states at inverse temperature $\beta$ \cite{ffs2024} in the $\beta \to \infty$ limit. The conditions contain a family of scalar inequalities indexed by observables of the system, and after truncating to a finite-dimensional subsspace of observables one can re-encode the family of scalar constraints by an SDP constraint \cite{ffs2024} in the $\beta \to \infty$ limit. (For general $\beta$ one can instead use a constraint which is convex but non-SDP, as done in \cite{ffs2024}.)

\section{Preliminaries}
\subsection{Noncommutative polynomial optimization problems}
A \textbf{noncommutative polynomial optimization problem} is defined by a finite set of noncommuting variables \[ X_1, \dots, X_n \]
along with a set of relations $q_{1}, \dots, q_{m}$, which are polynomials in the $X_{i}$ and $X_{i}^{*}$, together with a Hermitian polynomial \(p(X_1,\dots,X_n)\) satisfying $p^{*} = p$.
The goal is to maximize (or minimize) the expectation value of \(p\) over all Hilbert space representations of the variables that satisfy the relations, with respect to some quantum state \(\rho\). Concretely, the value of the problem is 
\[ \inf_{\rho, \{X_i\}} \mathrm{Tr}( \rho p(X_1,\dots,X_n) ) \]
where $\mathcal{H}$ is a fixed, infinite-dimensional separable Hilbert space, \(\rho\) is a density operator on \(\mathcal H\), and the supremum ranged over bounded operators \(X_i : \mathcal{H} \to \mathcal{H}\) such that $q_{j}(X_{1}, \dots, X_{n}) = 0$ for all $j \in [m]$. 

\subsubsection{The NPA hierarchy}
The NPA hierarchy \cite{npa2010a, npa2010b} gives a hierarchy of semidefinite programs which approximate the value of a noncommutative polynomial optimization problem. Specifically, one considers
\begin{align*}
\min \,\,\,&E[p] \\
\text{s.t.} \,\,\,&\begin{aligned}[t]
E &: \mathbb{R}\langle X_{1}, \dots, X_{n} \rangle_{\leq k} \to \mathbb{C} \\
0 &\leq E[q^{*}q] \qquad 2\deg q \leq k \\
0 &= E[q_{i}q] \qquad i \in [m], \deg q_{i} + \deg q \leq k
\end{aligned}
\end{align*}
where $k$ is an integer parameter and $E$ is defined on all noncommutative $*$-polynomials of degree at most $k$. In \cite{npa2010a}, it is proven that as $k \to \infty$, the value of these SDP converges monotonically to the true value of the noncommutative polynomial optimization problem.

\subsection{Algebras and Lie algebras defined by generators and relations}
The proof of asymptotic convergence will work with a Lie algebra defined by generators and relations, and show that it is isomorphic to $\mathfrak{su}(2^{n})$. Recall that a $*$-algebra $A$ is \textbf{free} on $n$ generators $x_{1}, \dots, x_{n}$ if it satisfies the following initiality (or universality) property: for any $n$ elements $b_{1}, \dots, b_{n}$ of a $*$-algebra $B$, there is a unique $*$-homomorphism $A \to B$ sending $x_{i}$ to $b_{i}$. Similarly, $A$ is defined by $n$ generators satisfying a set of relations if the same property holds for any elements $b_{1}, \dots, b_{n}$ which satisfy the same relations. To construct a $*$-algebra which is free, one can take the ring of noncommutative $*$-polynomials in a set of $n$ variables. To impose relations, one can take the quotient of the free $*$-algebra by a two-sided $*$-ideal generated by the given relations.

Similar versions of these constructions can be given for Lie algebras, which we will use later. The definition of quotient Lie algebras is relatively standard, but the construction of free Lie algebras is more involved. One can use \textit{Lyndon words} to construct an explicit basis, similar to the set of all words used for associative $*$-algebras \cite{r1993}. We will only use the existence of these constructions and the corresponding initiality/universality properties, and so do not give the full explicit construction here.

\subsection{Facts from representation theory}
The proof of convergence of our hierarchy uses a number of facts from representation theory, whose full proofs would be too lengthy to repeat here. Instead, we briefly state a number of the most important facts and refer the reader to the references \cite{b2013, dk2000, z1973} for more detail.

\begin{definition}
A (finite-dimensional) \textbf{representation} of a Lie group $G$ is a homomorphism 
\[
\pi: G \to GL(V)
\]
from $G$ into the group of invertible linear transformations on a complex vector space $V$. 
The dimension of $V$ is called the dimension of the representation. 
A representation is \emph{irreducible} (an \emph{irrep}) if $V$ has no nontrivial proper subspace which is invariant under all $\pi(g)$. 

A representation of a Lie algebra $\mathfrak{g}$ is a Lie algebra homomorphism 
\[
\rho: \mathfrak{g} \to \mathfrak{gl}(V)
\]
into the Lie algebra of endomorphisms of a vector space $V$, satisfying $\rho([X,Y]) = [\rho(X),\rho(Y)]$. 
Irreducibility is defined analogously. 
\end{definition}

\paragraph{Correspondence between representations of Lie groups and Lie algebras.}
Every finite-dimensional irreducible representation of a compact connected Lie group $G$, such as $SU(m)$, can be restricted to a Lie algebra $\mathfrak{g}$ by differentiating (in a suitable sense) the homomorphism $G \to GL(V)$, and remains irreducible. Conversely, every finite-dimensional irreducible representation of $\mathfrak{g}$ can be integrated uniquely to an irreducible representation of $G$. Thus we may freely identify irreps of $SU(m)$ with irreps of $\mathfrak{su}(m)$. Concretely, the Lie algebra $\mathfrak{su}(m)$ is the space of all skew-Hermitian matrices in $\mathbb{C}^{m \times m}$, and the Lie bracket is given by taking commutators.

\begin{definition}
If $(\pi,V)$ and $(\pi',W)$ are representations of a Lie group $G$, their tensor product is the representation 
\begin{align*}
\pi \otimes \pi' &:  G \to GL(V \otimes W) \\
(\pi \otimes \pi')(g) &= \pi(g)\otimes \pi'(g).
\end{align*}
Similarly, if $(\rho,V)$ and $(\rho',W)$ are representations of a Lie algebra $\mathfrak{g}$, their tensor product is
\begin{align*}
\rho \otimes \rho' &: \mathfrak{g} \to \mathfrak{gl}(V\otimes W) \\
(\rho \otimes \rho')(X) &= \rho(X)\otimes I_W + I_V\otimes \rho'(X).
\end{align*}
The tensor product is associative up to a canonical isomorphism.
\end{definition}

\begin{example}
Let $V=\mathbb{C}^m$ be the defining (fundamental) representation of $SU(m)$. 
For each $k=1,\dots,m$, the $k$-th exterior power
\[
\wedge^k V \coloneqq \{v \in V^{\otimes k} : \sigma \cdot v = \operatorname{sgn}(\sigma)v \}
\]
is a subspace of the $k$-fold tensor product $V^{\otimes k}$ consisting of antisymmetric tensors, where a permutation $\sigma \in S_{k}$ acts on $V^{\otimes k}$ by permuting indices.
The action of $SU(m)$ on $V^{\otimes k}$ is given by
\[ g\cdot (v_1 \otimes \cdots \otimes v_k) = (gv_1)\otimes \cdots \otimes (gv_k), \]
For the Lie algebra $\mathfrak{su}(m)$ one has
\[ X\cdot (v_1 \otimes \cdots \otimes v_k) = \sum_{i=1}^k v_1 \otimes \cdots \otimes (Xv_i) \otimes \cdots \otimes v_k. \]
It can be checked that the action is again well-defined on the space of antisymmetric tensors.
Each $\wedge^k V$ is an irreducible representation of $SU(m)$ with highest weight equal to the $k$-th fundamental weight. 
Its dimension is
\[ \dim(\wedge^k V) = \binom{m}{k}. \]
Some standard special cases include
\begin{itemize}
  \item the defining representation $\wedge^1 V \cong V$
  \item the two-index antisymmetric irrep $\wedge^2 V$
  \item the determinant representation $\wedge^m V \cong \mathbb{C}$, which is isomorphic to the trivial irrep for $SU(m)$.
\end{itemize}
\end{example}

\begin{remark}
If $H$ is a quantum Hamiltonian, where we think of $iH \in \mathfrak{su}(m)$, then the action of $H$ on the two-index antisymmetric irrep gives another Hamiltonian acting on a different space whose smallest eigenvalue is $\lambda_{0}(H) + \lambda_{1}(H)$, which is the sum of the two smallest eigenvalues of $H$.
\end{remark}

\begin{definition}
We denote by $\epsilon_{abc}$ the Levi--Civita symbol in three indices:
\[ \epsilon_{abc} = \begin{cases}
+1 & \text{if $(a,b,c)$ is an even permutation of $(1,2,3)$}, \\
-1 & \text{if $(a,b,c)$ is an odd permutation of $(1,2,3)$}, \\
0 & \text{if any two indices coincide}.
\end{cases} \]
The Lie bracket of $\mathfrak{su}(2)$ can be written as $[\sigma^a,\sigma^b] = 2i\,\epsilon_{abc}\sigma^c$.
\end{definition}

\paragraph{Schur--Weyl duality for \texorpdfstring{$SU(m)$}{SU(m)}.}
Let $V = \mathbb{C}^m$ be the defining representation of $SU(m)$. The tensor power $V^{\otimes n}$ carries commuting actions of $SU(m)$ and the symmetric group $S_n$, which have both been defined above. Schur--Weyl duality states that
\[ V^{\otimes n} \cong \bigoplus_{\substack{\lambda \vdash n \\\ell(\lambda)\leq m}} V_\lambda \otimes W_\lambda, \]
where the sum is over partitions $\lambda$ of $n$ with at most $m$ parts, $W_\lambda$ is the irreducible $SU(m)$ representation with highest weight $\lambda$, and $V_\lambda$ is the corresponding irreducible representation of $S_n$, both sets of which exhaust all irreducible representations of the two groups. We will sometimes denote the irreducible representation $W_{\lambda}$ by $S_{\lambda}(V)$.

\paragraph{Young tableau bases for irreducible representations of \texorpdfstring{\(SU(m)\)}{SU(m)}.}

Let \(\lambda=(\lambda_1,\lambda_2,\dots,\lambda_\ell)\) be a Young diagram of $m$ boxes. 
A \emph{semistandard Young tableau} (SSYT) of shape \(\lambda\) with entries in \(\{1,\dots,m\}\) is a filling of the boxes of \(\lambda\) with numbers from \(\{1,\dots,m\}\) such that the entries are weakly increasing along each row and strictly increasing down each column. The irreducible representation $W_{\lambda}$ has a distinguished basis indexed by semistandard Young tableaux of shape \(\lambda\). For a tableau \(T\), let
\[ c_i(T)\coloneqq \#\{\text{boxes of $T$ filled with $i$}\} \]
so that \(\sum_{i=1}^m c_i(T)=|\lambda|\). For \(g=\mathrm{diag}(t_1,\dots,t_m)\in GL(m,\mathbb C)\) we have
\[ g\cdot v_T = \left(\prod_{i=1}^m t_i^{c_i(T)}\right) v_T . \]
Similarly, for the Lie algebra action we have that for \(H=\mathrm{diag}(h_1,\dots,h_m)\in\mathfrak{gl}(m,\mathbb C)\),
\[ H\cdot v_T = \left(\sum_{i=1}^m c_i(T)h_i\right) v_T . \]
The weight of \(v_T\) is \(\mathrm{wt}(v_T)=\sum_{i=1}^m c_i(T)e_{i}\), where $e_{i}$ is the $i$-th standard basis vector. The highest weight of \(V^\lambda\) corresponds to the partition \(\lambda\) itself.

The Young tableau basis is constructed from a fixed ordered basis \((e_1,\dots,e_m)\) of the defining representation \(\mathbb C^m\). A box filled with \(i\) in a tableau corresponds \(e_i\). In particular, different choices of basis of $\mathbb{C}^{m}$ give different Young tableau bases of each \(V^\lambda\).

\paragraph{Casimir operators for \texorpdfstring{$SU(m)$}{SU(m)}.}
Let $\mathfrak{su}(m)$ have Killing form $\kappa(X,Y) = 2m\operatorname{Tr}(XY)$ on traceless matrices. 
Choose a basis $\{X_i\}$ of $\mathfrak{su}(m)$ and a dual basis $\{X^i\}$ with respect to $\kappa$. 
The quadratic Casimir element is
\[ C = \sum_i X_i X^i \in U(\mathfrak{su}(m)), \]
where $U(\mathfrak{su}(m))$ denotes the universal enveloping algebra. 
The Casimir element is in the center of $U(\mathfrak{su}(m))$ and hence acts as a scalar on every irreducible representation. 
If $V_\lambda$ is the irrep of $SU(m)$ (or $\mathfrak{su}(m)$) with highest weight $\lambda = (\lambda_1,\dots,\lambda_{m-1})$, then $C$ acts on $V_\lambda$ by
\[ C \cdot v = \left( \langle \lambda, \lambda + 2\rho \rangle \right) v \]
where $\rho$ is the Weyl vector, which is half the sum of positive roots, and $\langle \cdot,\cdot\rangle$ is the standard inner product on weights. 

\section{Upper bounds for ground energies of qubit Hamiltonians}\label{sec:upperbounds}
The main part of our hierarchy will give lower bounds $A \leq \lambda_{0}(H) + \lambda_{1}(H)$ for a local Hamiltonian $H$. To deduce bounds on the spectral gap, one then needs an \textbf{upper} bound $B \geq \lambda_{0}(H)$, from which we get
\[ \lambda_{1}(H) - \lambda_{0}(H) = \lambda_{0}(H) + \lambda_{1}(H) - 2\lambda_{0}(H) \geq A - 2B. \]
The simplest way to obtain such an upper bound $B$ would be to have some ansatz of states which are easy to optimize over with respect to the expectation value of $H$. Keeping with the spirit of this paper, we will describe an approach based on SDP hierarchies.

\subsection{The Lasserre hierarchy of upper bounds}\label{sec:upperbounds-lasserre}
The earliest such hierarchy is due to Lasserre \cite{l2011}, in the classical/commutative setting, and starts by fixing a reference measure $\mu$ on the feasible set with known moments. The ansatz of states considered is the set of probability distributions $p\,d\mu$, where $p$ ranges over low-degree sum-of-squares polynomials such that $\int p\,d\mu = 1$. Concretely, we have the sequence of upper bounds which can be defined as
\[f^{(d)} = \inf_{\substack{p\in \Sigma[x]\\\deg p\leq 2d}}  \frac{\int f(x)\,p(x)\,d\mu(x)}{\int p(x)\,d\mu(x)} \]
and computed using an SDP.
Knowledge of the moments of $\mu$ is required to be able to efficiently check the normalization for a given polynomial, as well as to be able to compute $\int fp\,d\mu$ where $f$ is the objective function.

\subsubsection{Noncommutative version of the Lasserre hierarchy of upper bounds}
In \cite{kmmv2025}, a noncommutative version of the Lasserre hierarchy of upper bounds was introduced.
Here one considers a unital \(*\)-algebra \(\mathcal A\) generated by noncommuting elements \(X_1,\dots,X_m\) with known relations, together with a fixed reference state \(\varphi\) on \(\mathcal A\). For degree bound \(2d\), one optimizes over noncommutative sums of squares \(P\in \Sigma\langle X \rangle\) of degree at most \(2d\), forming the reweighted state
\[ \psi_{P}(a) \coloneqq \frac{\varphi(P a)}{\varphi(P)}. \]
Given a Hermitian polynomial \(f\in\mathcal A\), the level-\(d\) bound is
\[ f^{(d)} := \inf_{\substack{P\in\Sigma\langle X\rangle\\ \deg P\leq 2d}} \frac{\varphi(P f)}{\varphi(P)}. \]
As in the commutative case, knowledge of the noncommutative $*$-moments of $\varphi$ is required to be able to efficiently check the normalization condition and compute the expectation of $f$, and the constraint \(P\in\Sigma\langle X\rangle\) is semidefinite representable.  
The sequence of bounds decreases with \(d\) and converges to the true minimum of \(f\) over all representations compatible with \(\varphi\) in sufficiently nice cases.

\subsubsection{Upper bounds for ground energies of qubit Hamiltonians}
For quantum local Hamiltonians, one can consider a straightforward noncommutative analog of the Lasserre hierarchy following the previous general outline. The algebra is the one $\mathcal{B}(\mathbb{C}^{2^{n}})$ of operators on $\mathbb{C}^{2^{n}}$, which is generated by Pauli matrices $X_{i}, Y_{i}, Z_{i}$ subject to standard commutation relation. The natural choice of a state is the maximally mixed state, for which the expectation of $A$ is given by $2^{-n}\operatorname{Tr}(A)$. The moments in terms of the Pauli matrices are all 0, except for the expectation of the identity operator (which is 1).

Approximate ground space projectors (AGSPs) are operators which are meant to approximate the true projector onto the ground space projector but have controlled complexity, and are often used in proofs of area laws for geometrically local Hamiltonians \cite{aklv2013, lvv2015, aag2022}.
In the SDP hierarchy for qudit systems, the reweighting polynomial \(P\) plays the role of such an operator. Specifically, the normalized state \(\psi(a)=\varphi(P a)/\varphi(P)\) concentrates weight near the ground space. As an operator on $\mathbb{C}^{2^{n}}$ it can be seen as an approximate ground space projector, and the degree of the polynomial can be seen as measuring its complexity.

\subsection{Hierarchy based on entropy-energy balance conditions}
As introduced by \cite{ffs2024} in the context of optimization hierarchies, one has the following characterization of a thermal state $\phi$ of a Hamiltonian $H$ at inverse temperature $\beta$:
\begin{align*}
\phi([H, a]) &= 0 \\
\phi(a^{*}[H, a]) &\geq \frac{1}{\beta}\phi(a^{*}a) \log \left( \frac{\phi(a^{*}a)}{\phi(aa^{*})} \right)
\end{align*}
where $a$ in both constraints ranges over all observables (here we only consider Hamiltonians on a finite number of qubits). In the limit as $\beta \to \infty$, which corresponds to the ground state of $H$, the second constraint becomes $\phi(a^{*}[H, a]) \geq 0$, which is equivalent to the matrix $[\phi(a_{i}^{*}[H, a_{j}])]_{i,j}$ being positive semidefinite for a linear basis $\{a_{i}\}$ of the space of observables.

When $H$ is local, one can then obtain an SDP hierarchy of outer approximations to the singleton set containing the ground state by only considering $a$ of bounded locality and replacing $\phi$ by a pseudo-state $\tilde{\phi}$ defined only on operators of bounded locality. By maximizing $\tilde{\phi}(H)$ subject to these constraints one can then obtain a hierarchy of upper bounds on the ground state energy of $H$. In some cases this can give more useful or tighter bounds compared to the noncommutative analog of the Lasserre hierarchy of upper bounds, as it will for the example we given in Section~\ref{sec:example}.

\section{Generators and relations for \texorpdfstring{$\mathfrak{su}(2^{n})$}{su(2\^n)}}\label{sec:genrel}
In this section, we give a set of generators of $\mathfrak{su}(2^{n})$ as a Lie algebra and determine a complete set of relations which these generators satisfy. The proof starts by first defining the generators, of which there are $O(n^{2})$. Then for each $k$-local Pauli operator $O \in \mathfrak{su}(2^{n})$ we define a canonical Lie polynomial in generators which is sent to $O$. We then define a set of relations in the generators such that any Lie polynomial in the generators is equivalent to a linear combination of the expressions corresponding to $k$-local Pauli operators defined previously. The proof amounts to a series of induction arguments which iteratively rewrite a given Lie polynomial into a normal form.

\begin{warning}
In this section (and only in this section), we adopt the convention that the Pauli matrices are given by
\begin{align*}
\sigma^{1} &= \frac{i}{2}X & \sigma^{2} &= \frac{i}{2}Y & \sigma^{3} &= \frac{i}{2}Z
\end{align*}
so that the commutation relations $[\sigma^{a}, \sigma^{b}]$ satisfy commutation relations only involving $\pm 1$, without additional factors of $i$ or 2. This simplifies the notation in our proof, which is already quite dense. Later sections defining the full hierarchy will use the usual convention for the Pauli matrices, while referring back to the Lie algebra relations defined here (which must be taken with the appropriate constant factors included).
\end{warning}

\subsection{Definition and set-up}
Let $n \geq 1$ be an integer.

\begin{definition}
Denote by $\mathfrak{f}_{n}$ the free Lie algebra generated by the symbols $S_{i}^{a}$ and $T_{ij}^{ab}$ for $1 \leq i < j \leq n$ and $a, b \in \{1, 2, 3\}$. We define a Lie algebra homomorphism $f_{n} : \mathfrak{f}_{n} \to \mathfrak{su}(2^{n})$ by
\begin{align*}
f_{n}(S_{i}^{a}) &= \sigma_{i}^{a} & f_{n}(T_{ij}^{ab}) &= \sigma_{i}^{a}\sigma_{j}^{b}.
\end{align*}
\end{definition}

Our goal in this section is to determine the kernel of the homomorphism $f_{n}$, or in other words to determine all Lie algebra relations satisfied by the images of the generators of $\mathfrak{f}_{n}$ in $\mathfrak{su}(2^{n})$. To this end, we will define a set of distinguished elements of $\mathfrak{f}_{n}$, along with a linear map $g_{n} : \mathfrak{su}(2^{n}) \to \mathfrak{f}_{n}$ which is a one-sided inverse of $f_{n}$, that is $f_{n} \circ g_{n} = \operatorname{id}_{\mathfrak{su}(2^{n})}$. After defining $g_{n}$, we will show that there is a set of relations of polynomial size such that the quotient of $\mathfrak{f}_{n}$ by these relations is isomorphic to $\mathfrak{su}(2^{n})$, with the isomorphisms being given by the induced maps of $f_{n}$ and $g_{n}$.

\begin{definition}
We define $A_{n}$ to be the set containing each $S_{i}^{a}$ for $1 \leq i \leq n$, and all expressions of the form \[ [[ \cdots [[S_{i_{1}}^{b_{1}}, T_{i_{1}i_{2}}^{a_{1}b_{2}}], T_{i_{2}i_{3}}^{a_{2}b_{3}}], \dots], T_{i_{m - 1}, i_{m}}^{a_{m - 1}b_{m}}] \]
where $i_{1} < i_{2} < \cdots < i_{m}$ and $b_{i}$ is always before $a_{i}$ in the cyclic order on $\{1, 2, 3\}$, or equivalently $(b_{i}, a_{i}) \in \{(1, 2), (2, 3), (3, 1)\}$. We define a linear map $g_{n} : \mathfrak{su}(2^{n}) \to \mathfrak{f}_{n}$ by first defining
\begin{align*}
g_{n}(\sigma_{i}^{a}) &= S_{i}^{a} \\
g_{n}(\sigma_{i_{1}}^{c_{1}} \cdots \sigma_{i_{m}}^{c_{m}}) &= [[ \cdots [[S_{i_{1}}^{b_{1}}, T_{i_{1}i_{2}}^{a_{1}b_{2}}], T_{i_{2}i_{3}}^{a_{2}b_{3}}], \dots], T_{i_{m - 1}, i_{m}}^{a_{m - 1}b_{m}}]
\end{align*}
for all $i_{1} < i_{2} < \cdots < i_{m}$, where $b_{m} = c_{m}$, and $(a_{j}, b_{j})$ is the unique element of $\{(1, 2), (2, 3), (3, 1)\}$ not containing $c_{j}$. The definition of $g_{n}$ is then extended linearly. We also define $g_{n}'$ similarly by $g_{n}'(\sigma_{i}^{a}) = S_{i}^{a}$ and
\[ g_{n}'(\sigma_{i_{1}}^{c_{1}} \cdots \sigma_{i_{m}}^{c_{m}}) = [[ \cdots [[T_{i_{1}i_{2}}^{c_{1}b_{2}}, T_{i_{2}i_{3}}^{a_{2}b_{3}}], T_{i_{3}i_{4}}^{a_{3}b_{4}}]], \dots], T_{i_{m - 1}, i_{m}}^{a_{m - 1}b_{m}}] \]
\end{definition}

For completeness, we include the following proof that $f_{n}$ is surjective. The proof is likely standard and is similar to the proof that quantum circuits built from 1- and 2-local unitary gates can approximate any $n$-qubit unitary.

\begin{proposition}\label{prop:1side-inv}
For any element $x \in \mathfrak{su}(2^{n})$, we have $f_{n}(g_{n}(x)) = x$ and $f_{n}(g_{n}'(x)) = x$.
\end{proposition}
\begin{proof}
By linearity, it suffices to check the statement for Pauli strings. We prove the statement by induction on the cardinality $m$ of the support of a given Pauli string. The base case $m = 1$ is clear. For the induction step, let $\sigma_{i_{1}}^{c_{1}} \cdots \sigma_{i_{m}}^{c_{m}}$ be a given Pauli string. Let $(b_{m - 1}, a_{m - 1}) \in \{(1, 2), (2, 3), (3, 1)\}$ such that neither are equal to $c_{m - 1}$. Applying the induction hypothesis to $\sigma_{i_{1}}^{c_{1}} \cdots \sigma_{i_{m-2}}^{c_{m-2}}\sigma_{i_{m - 1}}^{b_{m-1}}$ gives
\[
f_{n}([[ \cdots [[S_{i_{1}}^{b_{1}}, T_{i_{1}i_{2}}^{a_{1}b_{2}}], T_{i_{2}i_{3}}^{a_{2}b_{3}}], \dots], T_{i_{m - 2}, i_{m - 1}}^{a_{m - 2}b_{m - 1}}]) = \sigma_{i_{1}}^{c_{1}} \cdots \sigma_{i_{m-2}}^{c_{m-2}}\sigma_{i_{m - 1}}^{b_{m - 1}}
\]
so
\begin{align*}
&\phantom{{}={}} f_{n}( [[ \cdots [[S_{i_{1}}^{b_{1}}, T_{i_{1}i_{2}}^{a_{1}b_{2}}], T_{i_{2}i_{3}}^{a_{2}b_{3}}], \dots], T_{i_{m - 1}, i_{m}}^{a_{m - 1}c_{m}}]) \\
&= [f_{n}( [[ \cdots [[S_{i_{1}}^{b_{1}}, T_{i_{1}i_{2}}^{a_{1}b_{2}}], T_{i_{2}i_{3}}^{a_{2}b_{3}}], \dots], T_{i_{m - 2}, i_{m - 1}}^{a_{m - 2}b_{m - 1}}]), f_{n}(T_{i_{m-1},i_{m}}^{a_{m-1}c_{m}})] \\
&= [\sigma_{i_{1}}^{c_{1}} \cdots \sigma_{i_{m-2}}^{c_{m-2}}\sigma_{i_{m - 1}}^{b_{m - 1}}, \sigma_{i_{m-1}}^{a_{m-1}}\sigma_{i_{m}}^{c_{m}}] \\
&= \sigma_{i_{1}}^{c_{1}} \cdots \sigma_{i_{m-2}}^{c_{m-2}}[\sigma_{i_{m - 1}}^{b_{m - 1}}, \sigma_{i_{m-1}}^{a_{m-1}}] \sigma_{i_{m}}^{c_{m}} \\
&= \sigma_{i_{1}}^{c_{1}} \cdots \sigma_{i_{m}}^{c_{m}}
\end{align*}
finishing the induction step. A similar induction argument shows that $f_{n}(g_{n}'(x)) = x$.
\end{proof}

We can now use $g_{n}$ and $g_{n}'$ to define the relations we will use.

\begin{definition}
Define the following system of noncommutative polynomial equations in the generators of $\mathfrak{f}_{n}$, which is also a system of Lie polynomial equations:
\begin{subequations} \label{eq:rels}
\begin{align}
[S_{i}^{a}, S_{j}^{b}] &= \delta_{ij} \sum _{c} \epsilon_{abc} S_{i}^{c} \label{eq:11} \\
[S_{i}^{a}, T_{ij}^{bc}] &= \delta_{ij} \sum _{d} \epsilon_{abd} T_{jk}^{dc} + \delta_{ik} \sum _{d} \epsilon_{acd} T_{jk}^{dc} \label{eq:12} \\
[T_{ij}^{ab}, T_{kl}^{cd}] &= \begin{cases}
-[T_{kj}^{cb}, T_{il}^{ad}] & i = k, j \neq l \\
-[T_{ik}^{ac}, T_{jl}^{bd}] & j = k, i \neq l \\
-[T_{il}^{ad}, T_{kj}^{cb}] & j = l, i \neq k \\
\end{cases} \label{eq:22swap} \\
[T_{ij}^{ab}, T_{kl}^{cd}] &= \begin{cases}
\sum_{e} \epsilon_{bde} S_{j}^{e} & i = k, j = l, a = c, b \neq d \\
\sum_{e} \epsilon_{ace} S_{i}^{e} & i = k, j = l, a \neq c, b = d \\
0 & \text{otherwise}
\end{cases} \label{eq:22both} \\
[[T_{ij}^{ab}, T_{kl}^{cd}], T_{rs}^{ef}] &= g_{n}'(f_{n}([[T_{ij}^{ab}, T_{kl}^{cd}], T_{rs}^{ef}])) \label{eq:222} \\
[[[T_{ij}^{ab}, T_{kl}^{cd}], T_{rs}^{ef}], T_{tu}^{gh}] &= g_{n}'(f_{n}([[[T_{ij}^{ab}, T_{kl}^{cd}], T_{rs}^{ef}], T_{tu}^{gh}])) \label{eq:2222} \\
[[S_{i}^{a}, T_{jk}^{bc}], T_{lr}^{de}] &= g_{n}(f_{n}([[S_{i}^{a}, T_{jk}^{bc}], T_{lr}^{de}])). \label{eq:221}
\end{align}
\end{subequations}
\end{definition}

\begin{remark}
The first four families of relations can also be written as $x = g_{n}(f_{n}(x))$ or $x = g_{n}'(f_{n}(x))$ for a given element $x \in \mathfrak{f}_{n}$, but for clarity we have chosen to write them out explicitly.
\end{remark}

\begin{definition}
We denote by $\mathfrak{g}_{n}$ the Lie algebra which is the quotient of $\mathfrak{f}_{n}$ by the relations in Equation~\ref{eq:rels}.
\end{definition}

While the relations defined in Equation~\ref{eq:rels} may be difficult to read, the proof that the homomorphism $f_{n}$ is well-defined will help motivate where each of the relations comes from.

\begin{proposition}\label{prop:surjective}
The homomorphism $f_{n} : \mathfrak{f}_{n} \to \mathfrak{su}(2^{n})$  induces a well-defined homomorphism $\mathfrak{g}_{n} \to \mathfrak{su}(2^{n})$.
\end{proposition}
\begin{proof}
It suffices to check that $f_{n}$ sends both sides of each relations to the same element of $\mathfrak{su}(2^{n})$. The relation in Equation~\ref{eq:11} states that $[\sigma_{i}^{a}, \sigma_{j}^{b}]$ is 0 if $i \neq j$ and otherwise the appropriate Pauli operator on site $i$. The relation in Equation~\ref{eq:12} states that the commutator of a 1-local and a 2-local Pauli operator is 0 if they commute, and otherwise the other 2-local Pauli operator if the supports intersect each other.

The relation in Equation~\ref{eq:22swap} states that the commutator of two 2-local Pauli operators can be written as another commutator of two 2-local Pauli operators when the supports intersect on a single site; in this case we have
\[ [\sigma_{i}^{a}\sigma_{j}^{b}, \sigma_{j}^{c}\sigma_{k}^{d}] = \sigma_{i}^{a}[\sigma_{j}^{b}, \sigma_{j}^{c}]\sigma_{k}^{d} = -\sigma_{i}^{a}[\sigma_{j}^{c}, \sigma_{j}^{b}]\sigma_{k}^{d} = -[\sigma_{i}^{a}\sigma_{j}^{c}, \sigma_{j}^{b}\sigma_{k}^{d}] \]
and similarly for the other cases. When both supports are the same, as in Equation~\ref{eq:22both}, either the two operators commute or on one of the sites the operators are the same and on the other they are different; in the first case the commutator is zero, and in the second we have
\[ [\sigma_{i}^{a}\sigma_{j}^{b}, \sigma_{i}^{a}\sigma_{j}^{c}] = (\sigma_{i}^{a})^{2}\sigma_{j}^{b}\sigma_{j}^{c} - (\sigma_{i}^{a})^{2} \sigma_{j}^{c}\sigma_{j}^{b} = [\sigma_{j}^{b}, \sigma_{j}^{c}] \]
which is a 1-local Pauli operator.

For the final relation in Equation~\ref{eq:222}, we have $f_{n}(g_{n}'(f_{n}(x))) = f_{n}(x)$ by applying Proposition~\ref{prop:1side-inv}. (The same argument works for the other relations if they are written in a similar manner.)
\end{proof}

\subsection{Proof of completeness of the relations}

The following lemma is standard, but we prove it for completeness.

\begin{lemma}\label{lem:right}
If $x, y \in \mathfrak{f}_{n}$, then $[x, y]$ can be written as a linear combination of expressions of the form $[ \cdots [[x, z_{1}], z_{2}], \dots, z_{m}] $ where $z_{1}, \dots, z_{m}$ are one of the generators of $\mathfrak{f}_{n}$.
\end{lemma}
\begin{proof}
The element $y \in \mathfrak{f}_{n}$ can be written as a Lie polynomial in the generators, which is an expression constructed inductively as either one of the generators or a commutator of two such expressions. We use induction on the depth of the expression to prove the lemma. The base case is when $y$ itself is a generator, which holds vacuously. For the induction step, let $y = [y_{1}, y_{2}]$. Then the Jacobi identity gives
\begin{align*}
[x, y] &= [x, [y_{1}, y_{2}]] \\
&= [[x, y_{1}], y_{2}] - [[x, y_{2}], y_{1}].
\end{align*}
For the first term, the induction hypothesis allows us to write $[x, y_{1}]$ as a linear combination of terms of the form $x_{i} = [ \cdots [[x, z_{1}], z_{2}], \dots, z_{m'}]$. The induction hypothesis can then be applied to each $[x_{i}, y_{2}]$ and $y_{2}$, and taking an appropriate linear combination shows that $[[x, y_{1}], y_{2}]$ can be written as an expression of the desired form. A similar argument holds for $[[x, y_{2}], y_{1}]$.
\end{proof}

Our goal is to show that any element in $\mathfrak{g}_{n}$ is equal to a linear combination of elements of $A$ under the relations from Equation~\ref{eq:rels}. We will split the proof into several lemmas.

\begin{lemma}\label{lem:commute-down}
Let $x \in A$ be a given element written as \[ x = [[ \cdots [[S_{i_{1}}^{b_{1}}, T_{i_{1}i_{2}}^{a_{1}b_{2}}], T_{i_{2}i_{3}}^{a_{2}b_{3}}], \dots], T_{i_{m - 1}, i_{m}}^{a_{m - 1}b_{m}}]. \] Let $g$ be a given generator and $k$ be the largest index such that $i_{k}$ is equal to a subscript of $g$. If such a $k$ does not exist then in $\mathfrak{g}_{n}$ we have that $[x, g] = 0$ and otherwise
\[ [x, g] = [[ \cdots [[[[ \cdots [[S_{i_{1}}^{b_{1}}, T_{i_{1}i_{2}}^{a_{1}b_{2}}], T_{i_{2}i_{3}}^{a_{2}b_{3}}], \dots], T_{i_{k}, i_{k + 1}}^{a_{k}b_{k + 1}}], g], T_{i_{k + 1},i_{k + 2}}^{a_{k + 1},b_{k + 2}}], \dots], T_{i_{m-1},i_{m}}^{a_{m-1}b_{m}}]. \]
\end{lemma}
\begin{proof}
We use induction. The Jacobi rule shows for any elements $x, y, z$ that
\[ [x, [y, z]] = [[x, y], z] + [y, [x, z]]. \]
In particular, if $[x, y] = 0$ then $[x, [y, z]] = [y, [x, z]]$. From the relations in Equation~\ref{eq:rels} we see that if $g$ and $h$ are generators where the indices in the subscripts are disjoint, then $[g, h] = 0$. Inductively applying this equality then proves the lemma.
\end{proof}

\begin{lemma}\label{lem:comm-up}
Suppose a given expression $x$ contains only indices strictly less than some integer $M$ and that $i < j \leq k < l$ with $k > M$. Then
\[ [[x, y], T_{kl}^{cd}] = [x, [y, T_{kl}^{cd}]]. \]
\end{lemma}
\begin{proof}
We have
\begin{align*}
[[x, y], T_{kl}^{cd}] &= [[x, T_{kl}^{cd}], y] + [x, [y, T_{kl}^{cd}]] \\
&= [x, [y, T_{kl}^{cd}]]
\end{align*}
since the first term is zero by Lemma~\ref{lem:commute-down} and the assumption that $k > M$.
\end{proof}

\begin{lemma}\label{lem:outer-cyclic}
Suppose a given expression $x$ contains only indices strictly less than some integer $M$ and that $i < j \leq k < l$ with $k > M$. Then
\[ [[x, T_{ij}^{ab}], T_{jk}^{cd}] = -[[x, T_{ij}^{ac}], T_{jk}^{bd}] \]
\end{lemma}
\begin{proof}
Using Lemma~\ref{lem:comm-up}, we have
\begin{align*}
[[x, T_{ij}^{ab}], T_{jk}^{cd}] &= [x, [T_{ij}^{ab}, T_{jk}^{cd}]] \\
&= [x, -[T_{ij}^{ac}, T_{jk}^{bd}]] \\
&= -[[x, T_{ij}^{ac}], T_{jk}^{bd}]
\end{align*}
using Equation~\ref{eq:22swap} in the second line and Lemma~\ref{lem:comm-up} again in the third.
\end{proof}

\begin{lemma} \label{lem:comm1}
If $x \in A$, then $[x, S_{i}^{a}]$ is a linear combination of elements from $A$.
\end{lemma}
\begin{proof}
If $i$ does not appear in any subscript in the expression for $x$, then by Lemma~\ref{lem:commute-down} the commutator is 0. If it appears as the smallest index, then by Lemma~\ref{lem:commute-down} we can write
\begin{align*}
[x, S_{i}^{a}] &= [[ \cdots [[[S_{i_{1}}^{b_{1}}, S_{i_{1}}^{a}], T_{i_{1}i_{2}}^{a_{1}b_{2}}], T_{i_{2}i_{3}}^{a_{2}b_{3}}], \dots], T_{i_{m - 1}, i_{m}}^{a_{m - 1}b_{m}}].
\end{align*}
We can then apply Equation~\ref{eq:11} to simplify the innermost commutator. If the cyclic ordering condition does not hold for the resulting innermost commutator, we can apply Equation~\ref{eq:12} so that it does, possibly negating the expression.

If $i$ appears as the largest index, then similarly we can write $x = [x', T_{ji}^{bc}]$ and get
\begin{align*}
[x, S_{i}^{a}] &= [[x', T_{ji}^{bc}], S_{i}^{a}] \\
&= [[x', S_{i}^{a}], T_{ji}^{bc}] + [x', [T_{ji}^{bc}, S_{i}^{a}]] \\
&= [x', [T_{ji}^{bc}, S_{i}^{a}]]
\end{align*}
which can then be simplified using Equation~\ref{eq:12}. If it appears as some index which is neither the smallest nor largest, then we apply Lemma~\ref{lem:commute-down}. Working just with the part of the expression with the commutator involving $S_{i}^{a}$, we have a commutator of the form
\begin{align*}
&\phantom{{}={}} [[[x', T_{ji}^{bc}], T_{ik}^{de}], S_{i}^{a}] \\
&= [[x', T_{ji}^{bc}], [T_{ik}^{de}, S_{i}^{a}]] + [[x', [T_{ji}^{bc}, S_{i}^{a}]], T_{ik}^{de}] + [[[x', S_{i}^{a}], T_{ji}^{bc}], T_{ik}^{de}] \\
&= [[x', T_{ji}^{bc}], [T_{ik}^{de}, S_{i}^{a}]] + [[x', [T_{ji}^{bc}, S_{i}^{a}]], T_{ik}^{de}]
\end{align*}
applying Lemma~\ref{lem:commute-down} again. By applying Equation~\ref{eq:12} both of the resulting expressions can be simplified, and by Lemma~\ref{lem:outer-cyclic} we can guarantee the cyclic ordering condition holds. Since the bracket is bilinear, we can apply the same argument when the part of the expression outside of the commutator involving $S_{i}^{a}$ is included.
\end{proof}

\begin{lemma}\label{lem:comm10}
If $x \in A$ and $T_{ij}^{ab}$ is such that only $i$ and not $j$ appears in a subscript in the expression for $x$, then $[x, T_{ij}^{ab}]$ is a linear combination of elements of $A$.
\end{lemma}
\begin{proof}
If $i$ appears as the largest index, then the lemma follows directly. Thus we assume that $i$ does not appear as the largest index.

Applying Lemma~\ref{lem:commute-down}, we use induction on the number of commutators outside the position where $T_{ij}^{ab}$ was put. For the case when $T_{ij}^{ab}$ is put at the innermost level, we apply Equation~\ref{eq:221}. Otherwise, working just with the commutator involving $T_{ij}^{ab}$, we have an expression of the form
\begin{align*}
[[[x', T_{ki}^{cd}], T_{il}^{ef}], T_{ij}^{ab}] &= [[x', [T_{ki}^{cd}, T_{il}^{ef}]], T_{ij}^{ab}] \\
&= [x', [[T_{ki}^{cd}, T_{il}^{ef}], T_{ij}^{ab}]]
\end{align*}
by applying Lemma~\ref{lem:comm-up} twice. Applying Equation~\ref{eq:222}, the second factor in the commutator can be rewritten so that the indices in the subscripts appear in order. Then applying Lemma~\ref{lem:comm-up} twice again in reverse, we obtain an expression which is a linear combination of elements of $A$.

For the commutators at the outer levels, we do the following. Suppose using the notation of Lemma~\ref{lem:commute-down} that $T_{ij}^{ab}$ was put at level $k$. Then, after having applied the above rewriting, $T_{i_{k + 1},i_{k+2}}^{a_{k+1},b_{k+2}}$ has its first subscript included in the rewritten expression and not the second subscript. Concretely, including the commutator at level $k + 1$ we have either an expression of the form
\[ [[[[x', T_{ki}^{ab}], T_{ij}^{cd}], T_{jl}^{ef}], T_{lr}^{gh}] \]
if $j < l$, or
\[ [[[[x', T_{ki}^{ab}], T_{il}^{cd}], T_{lj}^{ef}], T_{lr}^{gh}] \]
if $j > l$, where $l, r, g, h$ are chosen such that $T_{lr}^{gh} = T_{i_{k + 1},i_{k+2}}^{a_{k+1},b_{k+2}}$.
Thus the same rewriting can be applied inductively to $T_{i_{k + 1},i_{k+2}}^{a_{k+1},b_{k+2}}$ and the rewritten expression, with the base case being that $T_{ij}^{ab}$ was put at the outermost level to begin with.
\end{proof}

\begin{lemma}\label{lem:comm01}
If $x \in A$ and $T_{ij}^{ab}$ is such that only $j$ and not $i$ appears in a subscript in the expression for $x$, then $[x, T_{ij}^{ab}]$ is a linear combination of elements of $A$.
\end{lemma}
\begin{proof}
Let $k$ be such that $i_{k} < i < i_{k + 1}$. We use induction on the number of indices in between $i_{k}$ and the position where $j$ appears in the expression for $x$. For the case when $j$ is already the smallest index appearing in $x$, we can apply Equation~\ref{eq:221}. Otherwise we have an expression of the form
\begin{align*}
[[[x', T_{rj}^{ab}], T_{jl}^{cd}], T_{ij}^{ef}] &= [[x', [T_{rj}^{ab}, T_{jl}^{cd}]], T_{ij}^{ef}] \\
&= [x', [[T_{rj}^{ab}, T_{jl}^{cd}], T_{ij}^{ef}]]
\end{align*}
by applying Lemma~\ref{lem:comm-up} twice. Applying Equation~\ref{eq:222}, the second factor can be rewritten so that the indices in the subscripts appear in order. Then applying Lemma~\ref{lem:comm-up} twice again in reverse, we get an expression of the form
\[ [[[x', T_{ir}^{ab}], T_{rj}^{cd}], T_{jl}^{ef}] \]
and can then apply the induction hypothesis to $T_{ir}^{ab}$ and $x'$. The base case is when the commutator involving $j$ already included $i$, in which case one can apply Equation~\ref{eq:22both}.

For the outer commutators which are outside the position where index $j$ appears, applying Lemma~\ref{lem:commute-down} and bilinearity of the Lie bracket then finishes the proof of the lemma.
\end{proof}

\begin{lemma}\label{lem:comm11}
If $x \in A$ and $T_{ij}^{ab}$ is such that both $i$ and $j$ appear in a subscript in the expression for $x$, then $[x, T_{ij}^{ab}]$ is a linear combination of elements of $A$.
\end{lemma}
\begin{proof}
We use induction on the number of indices in between $i$ and $j$ which appear in the expression for $x$. For the base case when there are no such indices, using Lemma~\ref{lem:commute-down} to work just with the relevant inner part gives an expression of the form
\begin{align*}
[[[[x', T_{ki}^{cd}], T_{ij}^{ef}], T_{jl}^{gh}], T_{ij}^{ab}] &= [[[x', [T_{ki}^{cd}, T_{ij}^{ef}]], T_{jl}^{gh}], T_{ij}^{ab}] \\
&= [[x', [[T_{ki}^{cd}, T_{ij}^{ef}], T_{jl}^{gh}]], T_{ij}^{ab}] \\
&= [x', [[[T_{ki}^{cd}, T_{ij}^{ef}], T_{jl}^{gh}], T_{ij}^{ab}]]
\end{align*}
applying Lemma~\ref{lem:comm-up} three times. Then using Equation~\ref{eq:2222} and applying Lemma~\ref{lem:comm-up} three times again in reverse gives an expression which is a linear combination of elements of $A$.

For the induction step, we use Lemma~\ref{lem:comm-up} twice to combine the commutators where $i$ and $j$ appear, obtaining an expression of the form
\[ [[ \cdots [[x', [T_{ki}^{cd}, T_{il}^{ef}]], \cdots], [T_{rj}^{gh}, T_{js}^{\alpha \beta}]] \]
ignoring the parts of the expression involving indices larger than $j$. By the Jacobi rule and Lemma~\ref{lem:commute-down}, taking the commutator of this expression with $T_{ij}^{ab}$ gives an expression of the form
\[ [ \cdots [x', [[T_{ki}^{cd}, T_{il}^{ef}], T_{ij}^{ab}]], \cdots ] + [ \cdots [[x', [T_{ki}^{cd}, T_{il}^{ef}]], \cdots], [[T_{rj}^{gh}, T_{js}^{\alpha \beta}], T_{ij}^{ab}]]. \]
We inductively rewrite the two terms separately. For the second, we again use Equation~\ref{eq:222} similarly as in the induction steps of Lemma~\ref{lem:comm10} and Lemma~\ref{lem:comm01}, along with Lemma~\ref{lem:comm-up} used in reverse twice to obtain an expression of the form
\[ [[[[ \cdots [x', [T_{ki}^{cd}, T_{il}^{ef}]], \cdots], T_{ir}^{ab}], T_{rj}^{cd}], T_{js}^{ef}] \]
and then the induction hypothesis can be applied to the inner expression and $T_{ir}^{ab}$. For the first term, we apply a similar argument to obtain an expression of the form
\[ [[[x', T_{ki}^{ab}], T_{il}^{ef}], T_{lj}^{gh}] \]
in the inner part, after which the same argument can be applied inductively with $T_{lj}^{gh}$ and the rest of the expression on the outside.
\end{proof}

\begin{proposition}
The kernel of the homomorphism $\mathfrak{g}_{n} \to \mathfrak{su}(2^{n})$ induced by $f_{n}$ is 0.
\end{proposition}
\begin{proof}
From Proposition~\ref{prop:surjective}, we know that the induced map is surjective. Thus it suffices to show that the dimension of $\mathfrak{g}_{n}$ as a vector space is finite and bounded above by $\dim \mathfrak{su}(2^{n}) = 4^{n} - 1$. To this end, we claim that any element of $\mathfrak{g}_{n}$ is a linear combination of elements from the set $A$, each of which are written in a particular normal form. Concretely, each element of $A$ is the image of a Pauli operator under $g_{n}$, and is sent to an $m$-local Pauli operator under $f_{n}$.

As in Lemma~\ref{lem:right}, we use induction on the depth of the Lie polynomials defining the given elements. The base case is when the expression is one of the generators. For the induction step, we consider a commutator of the form $[x, y]$. The induction hypothesis implies that $x$ and $y$ are equal to linear combinations of elements of $A$. By linearity, it then suffices to consider the commutator of two elements of $A$. We proceed by cases on the type of element of $A$. The commutator of two generators of the form $S_{i}^{a}$ and $S_{j}^{b}$ is another such generator or 0.

For the second case, we consider the commutator of $S_{i}^{a}$ and an element of $A$. Then Lemma~\ref{lem:comm1} shows that this is a linear combination of elements of $A$.

Finally, we consider the commutator of two elements of $A$ which are both not generators. Then Lemma~\ref{lem:right} implies that is suffices to consider the commutator of an element of $A$ with a generator $T_{ij}^{ab}$ or $S_{i}^{a}$. The second case follows from Lemma~\ref{lem:comm1} as before, and the first follows from either Lemma~\ref{lem:commute-down}, Lemma~\ref{lem:comm01}, Lemma~\ref{lem:comm10}, or Lemma~\ref{lem:comm11} depending on whether $i$ or $j$ are included in the subscripts in the element of $A$ or not.
\end{proof}

Combining the two propositions gives the following theorem.

\begin{theorem} \label{thm:genrel}
The homomorphism $f_{n}$ induces an isomorphism of Lie algebras $\mathfrak{g}_{n} \cong \mathfrak{su}(2^{n})$.
\end{theorem}

\section{The hierarchy}\label{sec:hierarchy}
We will define an SDP hierarchy of lower bounds on $\lambda_{0}(H) + \lambda_{1}(H)$ as the NPA hierarchy associated to a given system of noncommutative polynomial equations. The set of variables is
\[ V_{n} \coloneqq \{S_{i}^{a} : i \in [n], a \in [3]\} \cup \{T_{ij}^{ab} : i, j \in [n], a, b \in [3]\} \]
subject to the noncommutative constraints
\begin{subequations} \label{eq:sys}
\begin{align}
&\text{all Lie algebra relations from Equation~\ref{eq:rels}} \\
(S_{i}^{a})^{*} &= S_{i}^{a} \\
0 &= S_{j}^{a}(S_{j}^{a} - 2)(S_{j}^{a} + 2) \label{eq:sys-eig} \\
C_{4} &\leq \sum _{a, b \in [3]} (T_{ij}^{ab})^{*}T_{ij}^{ab} \leq C_{4}' \label{eq:sys-eig2} \\
8(n - 1) &\geq \sum _{i = 1} ^{n} \sum _{a = 1} ^{3} (S_{i}^{a})^{*}S_{i}^{a} \label{eq:casimir-bound}
\end{align}
\end{subequations}
where $C_{4}$ is the eigenvalue of the quadratic Casimir operator of $SU(4)$ acting on the irreducible representation $\wedge^{2}(\mathbb{C}^{4})$ and $C_{4}'$ is the Casimir eigenvalue for $\operatorname{Sym}^{2}(\mathbb{C}^{4})$.

\subsection{Encoding a higher-locality Hamiltonian}
Encoding a $k$-local Hamiltonian for $k \geq 3$ is somewhat more nontrivial than in the usual case. For computing ground energies of local Hamiltonians, one typically thinks of $\mathcal{B}(\mathbb{C}^{2^{n}})$ as being generated by 1-local Pauli operators satisfying certain relations, and then a $k$-local Hamiltonian can be written as a degree-$k$ noncommutative polynomial in the 1-local Pauli generators.

In our hierarchy, while there do exist generators $S_{i}^{a}$ corresponding in some sense to 1-local Pauli operators, thinking of a $k$-local Hamiltonian $H$ as a polynomial in the $S_{i}^{a}$ will not give the correct result. The final goal is to bound the action $P_{\wedge^{2}(\mathbb{C}^{2^{n}})}(H \otimes I + I \otimes H)P_{\wedge^{2}(\mathbb{C}^{2^{n}})}$ of $H$ on the antisymmetric subspace, which is general not equal to a noncommutative polynomial in the $S_{i}^{a}$ even if $H$ is. One instead needs to write $H$ as a linear combination of Pauli terms, and then use either the map $g_{n}$ or $g_{n}'$ to write each Pauli term as an expression involving commutators in the $S_{i}^{a}$ and $T_{ij}^{ab}$. For Hamiltonians with only 1- and 2-local terms, which are some of the most common and important examples, the resulting expressions end up just being linear combinations of the generators $S_{i}^{a}$ and $T_{ij}^{ab}$.

More conceptually, the action of $H$ on the antisymmetric subspace is given by a representation of the Lie algebra $\mathfrak{su}(2^{n})$. In order to write this in a ``low-degree'' way, one should write $H$ as an element of $\mathfrak{su}(2^{n})$ using its generators and the Lie bracket, which when combined with the other relations implies that one gets the corresponding action of $H$ on the antisymmetric subspace. In contrast, writing it as a noncommutative polynomial in the $S_{i}^{a}$ treats the underlying algebra as the $*$-algebra of operators on $\mathbb{C}^{2^{n}}$, which does not give the Lie algebra action on the antisymmetric subspace.

\subsection{Proof of asymptotic convergence}
To prove asymptotic convergence of the hierarchy, it suffices (by asymptotic convergence of the general NPA hierarchy) to consider irreducible representations of the $*$-algebra defined by the generators in $V$ and corresponding relations. The proof will consist of using the Lie algebra relations to show that any irreducible representation of this algebra is also canonically an irreducible representation of the Lie algebra $\mathfrak{su}(2^{n})$, and then using the following propositions to deduce which irreducible representation it must be.

\begin{proposition}\label{prop:deg2-plethysm}
Consider $\mathbb{C}^{2^{n}}$ as the tensor product of the standard representations of $n$ copies of $SU(2)$. Then we have
\begin{align*}
\operatorname{Sym}^{2}(\mathbb{C}^{2^{n}}) &\cong \bigoplus _{\substack{v \in \{0, 2\}^{n}\\|v| = n \mod 2}} \bigboxtimes _{i = 1} ^{n} \operatorname{Sym}^{v_{i}}(\mathbb{C}^{2}) \\ 
\wedge^{2}(\mathbb{C}^{2^{n}}) &\cong \bigoplus _{\substack{v \in \{0, 2\}^{n}\\|v| = n - 1 \mod 2}} \bigboxtimes _{i = 1} ^{n} \operatorname{Sym}^{v_{i}}(\mathbb{C}^{2})
\end{align*}
as representations of $SU(2)^{n}$, where $|v|$ denotes the $L^{1}$-norm of an integer vector, and $\boxtimes _{i = 1} ^{n} V_{i}$ denotes the $SU(2)^{n}$ representation which is the tensor product of $n$ representations $V_{i}$ of $SU(2)$.
\end{proposition}
\begin{proof}
By Schur-Weyl duality, we have
\[ V^{\otimes 2} \cong 1_{S_{2}} \otimes \operatorname{Sym}^{2}(V) \oplus \operatorname{sgn}_{S_{2}} \wedge^{2}(V) \]
where $1_{S_{2}}$ and $\operatorname{sgn}_{S_{2}}$ denote the trivial and sign representations of $S_{2}$. When $V \cong \mathbb{C}^{2}$ we have $\wedge^{2}(V) \cong 1 \cong \operatorname{Sym}^{0}(\mathbb{C}^{2})$, the trivial representation of $SU(2)$. This gives
\begin{align*}
&\phantom{{}\cong{}} \operatorname{Sym}^{2}(\mathbb{C}^{2^{n}}) \\
&\cong \left(\left(\bigboxtimes _{i = 1} ^{n} \mathbb{C}^{2}\right)^{\otimes 2} \otimes 1_{S_{2}}\right)^{S_{2}} \\
&\cong \left( 1_{S_{2}} \otimes \bigboxtimes _{i = 1} ^{n} (\mathbb{C}^{2})^{\otimes 2} \right)^{S_{2}} \\
&\cong \left( 1_{S_{2}} \otimes \bigboxtimes _{i = 1} ^{n} (1_{S_{2}} \otimes \operatorname{Sym}^{2}(\mathbb{C}^{2}) \oplus \operatorname{sgn}_{S_{2}} \otimes \operatorname{Sym}^{0}(\mathbb{C}^{2})) \right)^{S_{2}} \\
&\cong \left(\bigoplus _{v \in \{0, 2\}^{n}} \operatorname{sgn}_{S_{2}}^{|v|/2} \otimes \bigboxtimes _{i = 1} ^{n} \operatorname{Sym}^{v_{i}}(\mathbb{C}^{2})\right)^{S_{2}} \\
&\cong \bigoplus _{\substack{v \in \{0, 2\}^{n}\\|v| = n \mod 2}} \bigboxtimes _{i = 1} ^{n} \operatorname{Sym}^{v_{i}}(\mathbb{C}^{2}).
\end{align*}
Similarly, we have
\begin{align*}
\wedge^{2}(\mathbb{C}^{2^{n}}) &\cong \left( \operatorname{sgn}_{S_{2}} \otimes \bigboxtimes _{i = 1} ^{n} (1_{S_{2}} \otimes \operatorname{Sym}^{2}(\mathbb{C}^{2}) \oplus \operatorname{sgn}_{S_{2}} \otimes \operatorname{Sym}^{0}(\mathbb{C}^{2})) \right)^{S_{2}} \\
&\cong \bigoplus _{\substack{v \in \{0, 2\}^{n}\\|v| = n - 1 \mod 2}} \bigboxtimes _{i = 1} ^{n} \operatorname{Sym}^{v_{i}}(\mathbb{C}^{2}).
\end{align*}
\end{proof}

\begin{theorem}\label{thm:res-reps}
Suppose $\mathcal{H}$ is a Hilbert space with a unitary action of the Lie group $SU(2^{n})$. Consider the associated action of the Lie algebra $\mathfrak{su}(2^{n})$ and the restriction to the 1- and 2-local Pauli operators in the image of $f_{n}$. If the corresponding actions of the generators in $V_{n}$ satisfy Equation~\ref{eq:sys}, then $\mathcal{H} \cong \mathcal{H}' \otimes \wedge^{2}(\mathbb{C}^{2^{n}})$ as representations of $SU(2^{n})$, where $\mathcal{H}'$ has the trivial action.
\end{theorem}
\begin{proof}
We first consider the case where $\mathcal{H}$ is an irreducible representation of $SU(2^{n})$. It suffices to prove the result in this case, since otherwise we can decompose $\mathcal{H}$ into irreducible representations and then deduce that no others can occur besides $\wedge^{2}(\mathbb{C}^{2^{n}})$. Thus let $\lambda$ be a partition with at most $2^{n} - 1$ rows indexing the irreducible representation isomorphic to $\mathcal{H}$.

For notational convenience, let $m = 2^{n}$. We first claim that Equation~\ref{eq:sys-eig} implies that $\lambda$ is either $(2)$, $(1, 1)$, $(1^{m - 2})$, $(2, 1^{m - 2})$, or $(2^{m - 1})$. (The notation $1^{k}$ or $2^{k}$ indicates that the corresponding partition has $k$ rows of size 1 or 2, respectively.) To show this, let $\lambda$ have $C$ columns with $c_{1}, \dots, c_{C}$ boxes.
On $\mathbb{C}^{m}$ the eigenvalues of $S_{i}^{a}$ are $\pm 1$ each with multiplicity $m/2$. The irreducible representation $S_{\lambda}(\mathbb{C}^{m})$ has a basis indexed by semistandard Young tableaux $T$ of shape $\lambda$ and maximum entry $2^{n}$, where each basis element is an eigenvector of the action of $S_{i}^{a}$ and the corresponding eigenvalue is $\sum _{i \in T} \lambda_{i}$, where $\lambda_{i}$ is the $i$-th eigenvalue of $S_{i}^{a}$ acting on $\mathbb{C}^{m}$. In $T$ the entries in each column must be strictly increasing and thus distinct, so the number of entries in column $i$ contributing $+1$ is at most $\min(c_{i}, m/2)$ and the contribution to a possible eigenvalue of $T$ is at most
\[ \min(c_{i}, m/2) - (c_{i} - \min(c_{i}, m/2)) = 2\min(c_{i}, m/2) - c_{i}. \]
Moreover, one can see that such a contribution is always possible by taking the entries in each column $i$ to be consecutive integers from $1$ to $c_{i}$. In particular, if there are $C \geq 3$ or more columns then there is an eigenvalue at least 3, contradicting Equation~\ref{eq:sys-eig}, so there are at most 2 columns. For the maximum eigenvalue to be exactly 2 we must have that either one column contributes 2 to the sum, which means $\lambda$ is either $(1, 1)$ or $(1^{m - 2})$, or there are two columns each contributing 1, so $\lambda$ is either $(2)$, $(2, 1^{m - 2})$, or $(2^{m - 1})$.

We now claim that Equation~\ref{eq:casimir-bound} implies that $\lambda$ cannot be $(2)$ or $(2^{m - 1})$. The highest weights of the two corresponding representations of $SU(m)$ differ by a vector whose entries are all $m$, so they give dual representations of $SU(m)$. The restrictions to $SU(2)^{n}$ are then isomorphic, and since Equation~\ref{eq:casimir-bound} only depends on the $SU(2)^{n}$ (or equivalently $\mathfrak{su}(2)^{\oplus n}$) action it suffices to just consider $(2)$.
The right-hand side of Equation~\ref{eq:casimir-bound} acts as the scalar $\sum _{i = 1} ^{n} v_{i}/2$ on the $SU(2)^{n}$ irrep $\boxtimes _{i = 1} ^{n} \operatorname{Sym}^{v_{i}}(\mathbb{C}^{2})$, so by Proposition~\ref{prop:deg2-plethysm} it will have an eigenvalue of $n$ on $S_{(2)}(\mathbb{C}^{m}) \cong \operatorname{Sym}^{2}(\mathbb{C}^{m})$.

We claim that a similar argument shows that $\lambda$ cannot be $(2, 1^{m - 2})$. It suffices to show that $S_{(2, 1^{m - 2})}(\mathbb{C}^{m})$ contains a copy of $\boxtimes _{i = 1} ^{n} \operatorname{Sym}^{2}(\mathbb{C}^{2})$. To prove this, it suffices to give a vector which is a simultaneous $+1$ eigenvector of $S_{i}^{1}$ for all $i \in [n]$, which is then a corresponding highest weight vector. Since the $S_{i}^{1}$ are simultaneously diagonalizable on $\mathbb{C}^{m}$,  we can take the basis element indexed by the Young tableau containing all elements but the simultaneous $-1$ eigenvector of the $S_{i}^{1}$ in the first column, and in the second column contains the simultaneous $+1$ eigenvector.

We finally claim that Equation~\ref{eq:sys-eig2} implies that $\lambda$ cannot be $(1^{m - 2})$. To show this, we consider $S_{(1^{m - 2})}(\mathbb{C}^{m}) \cong \wedge^{m - 2}(\mathbb{C}^{m})$ as a representation of $SU(4)^{n_{1}} \times SU(2)^{n_{0}}$, where $n = 2n_{1} + n_{0}$ and $n_{0} \in \{0, 1\}$. A similar argument as in the previous paragraph shows that it contains irreps where one of the tensor factors is an $SU(4)$ irrep with degree strictly greater than 2, implying that the Casimir operator of $SU(4)$ acts as a scalar with constant strictly greater than $C_{4}$ on that subspace, and so $\lambda$ cannot be $(1^{m - 2})$.
\end{proof}

The final step in the proof of asymptotic convergence is given by the following proposition.

\begin{proposition}
Any irreducible representation of the algebra generated by the variables $V_{n}$ satisfying the relations in Equation~\ref{eq:sys} is isomorphic to $\wedge^{2}(\mathbb{C}^{2^{n}})$, where the action of each of the variables $S_{i}^{a}$ and $T_{ij}^{ab}$ comes from the map $f_{n} : \mathfrak{f}_{n} \to \mathfrak{su}(2^{n})$ and the Lie algebra action of $\mathfrak{su}(2^{n})$ on $\wedge^{2}(\mathbb{C}^{2^{n}})$.
\end{proposition}
\begin{proof}
Let $\mathcal{H}$ be a Hilbert space and identify each variable in $V_{n}$ with an operator acting on $\mathcal{H}$, which together satisfy the constraints of Equation~\ref{eq:sys}. Then by Theorem~\ref{thm:genrel}, there exists a unique Lie algebra homomorphism $\rho : \mathfrak{su}(2^{n}) \to \mathcal{H}$ such that $\rho \circ f$ sends each variable to its original action on $\mathcal{H}$ and such that $\rho(g)$ is anti-self-adjoint for each $g \in \mathfrak{su}(2^{n})$. Such irreps are in bijection with irreps of $SU(2^{n})$, and the action of each generator $S_{i}^{a}$ or $T_{ij}^{ab}$ is the corresponding Lie algebra action. The action of any Lie polynomial in the generators is again the Lie algebra action, since the generators generate all of $\mathfrak{su}(2^{n})$. By Proposition~\ref{thm:res-reps}, the additional constraints force the irrep of $\mathfrak{su}(2^{n})$ to be the antisymmetric subspace $\wedge^{2}(\mathbb{C}^{2^{n}})$.
\end{proof}

\begin{remark}[Finite convergence]
While above we have applied the standard proof of asymptotic convergence of the NPA hierarchy, it is in fact true that the hierarchy converges at a finite level. One can see this by noting that generating $\mathfrak{su}(2^{n})$ using our given generating set uses only commutators of depth at most $n$, and such commutators expand to noncommutative polynomials of degree at most $n$. In particular, a linear basis for the $*$-algebra defined by our generators is given by such commutator expressions, and standard argument then imply that the hierarchy must converge at level $n$.
\end{remark}

\section{Examples of certificates}\label{sec:example}
\subsection{Commuting 1-local Hamiltonians}
In this section, we show that the hierarchy can certify a nontrivial lower bound on the spectral gap of the 1-local frustration-free Hamiltonian
\[ H_{n} \coloneqq -\sum _{i = 1} ^{n} Z_{i}. \]
Of course, this particular Hamiltonian is explicitly diagonalizable on the standard basis, with any standard basis vector of Hamming weight $k$ having eigenvalue $-n + 2k$, and the spectral gap is easily seen to be 2. We use it simply as an example where the hierarchy gives a nontrivial (that is, nonzero) lower bound on the spectral gap.

First, note that the constraint $Z_{i}(Z_{i} - 2)(Z_{i} + 2) = 0$ implies that any univariate polynomial in a single $Z_{i}$ is a degree-4 SOS up to the given relations if it is positive on the set $\{-2, 0, 2\}$, using standard facts about univariate polynomials. In particular, $Z_{i}^{*}Z_{i} - 2Z_{i}$ is a degree-4 SOS.

Next, we have
\begin{align*}
(X_{j} + iY_{j})^{*}(X_{j} + iY_{j}) &= X_{j}^{*}X_{j} + Y_{j}^{*}Y_{j} + iX_{j}^{*}Y_{j} - iY_{j}^{*}X_{j} \\
&= X_{j}^{*}X_{j} + Y_{j}^{*}Y_{j} + i[X_{j}, Y_{j}] \\
&= X_{j}^{*}X_{j} + Y_{j}^{*}Y_{j} - 2Z_{j}
\end{align*}
again using the given relations, thus the right-hand side is a degree-2 SOS. (Very pedantically, we have
\[ X_{j}^{*}X_{j} + Y_{j}^{*}Y_{j} - 2Z_{j} = (X_{j} + iY_{j})^{*}(X_{j} + iY_{j}) + i(X_{j} - X_{j}^{*})Y_{j} - i(Y_{j} - Y_{j}^{*})X_{j} + i(2iZ_{j} - [X_{j}, Y_{j}]) \]
as an explicit SOS expression for the right-hand side, using the given relations.) Summing the two inequalities over all $n$ gives that
\begin{align*}
&\phantom{{}={}} 8(n - 1) + 4H_{n} \\
&= 8(n - 1) - 4\sum _{j = 1} ^{n} Z_{j} \\
&= 8(n - 1) - \sum _{j = 1} ^{n} (X_{j}^{*}X_{j} + Y_{j}^{*}Y_{j} + Z_{j}^{*}Z_{j}) + \sum _{j = 1} ^{n} (X_{j}^{*}X_{j} + Y_{j}^{*}Y_{j} - 2Z_{j}) + \sum _{j = 1} ^{n} (Z_{j}^{*}Z_{j} - 2Z_{j})
\end{align*}
is a degree-4 SOS. In particular, the hierarchy certifies that the action of $H_{n}$ on the antisymmetric subspace has lowest eigenvalue at least $-2(n - 1)$, matching the true value of $-n - (n - 2)$.

To completely certify the bound, one should also bound the ground energy of $H_{n}$ from above. There are a few options for doing this. The Lasserre hierarchy of upper bounds seems to have particularly slow convergence in this example, but at a sufficiently high level will certify a tight enough upper bound to obtain a nontrivial lower bound on the spectral gap. If one is willing to assume an exact upper bound for frustration-free Hamiltonians, as is done in \cite{rketsb2024}, then combining it with the previous argument shows that the degree-4 relaxation certifies an exact lower bound on the spectral gap.

Alternatively, one can use the hierarchy of outer approximations to the singleton containing the ground state described in \cite{ffs2024}. Concretely, it is known that the hierarchy converges exactly at a finite level depending on the number of qubits. Since the given Hamiltonian is a sum of 1-local terms, taking a constant finite level of the hierarchy will give both exact upper and lower bounds on the ground state energy, which then gives an exact lower bound on the spectral gap when combined with the previous argument.

\subsection{Transverse-field \texorpdfstring{$XY$}{XY} model on regular graphs}\label{sec:transverse}
The operator $\sum _{i = 1} ^{n} Z_{i}$ is the transverse-field term in models such as the transverse-field Ising model. Here we consider the Hamiltonian
\[ H_{G}(h, J) \coloneqq -h\sum _{i = 1} ^{n} Z_{i} - \frac{J}{2} \sum _{(i, j) \in E(G)} (X_{i}X_{j} + Y_{i}Y_{j}) \]
for the isotropic $XY$ model with a transverse field. Of course, when $J = 0$ this is a scalar multiple of the commuting 1-local Hamiltonian considered before. When $G$ is a cycle graph the $XY$ model was shown to have an exact solution via free fermions by Lieb, Schultz, and Mattis \cite{lsm1961}. The transverse-field $XY$ model on a cycle graph also has an exact solution by free fermions, first found by Katsura \cite{k1962}, from which one can bound the spectral gap by $2(h - J)$. The model was also studied on paths and complete graphs by Osborne \cite{o2006}. In this section we will prove a lower bound of $2h - dJ$ on the spectral gap of $H_{G}(h, j)$ whenever $G$ is a $d$-regular graph, recovering the bound of $2(h - J)$ when $G$ is a cycle graph.

First, we give a straightforward upper bound on the ground energy of $H_{G}(h, J)$.
\begin{proposition}
At level 1, the entropy-energy balance hierarchy certifies
\[ \lambda_{0}(H_{G}(h, J)) \leq -hn. \]
\end{proposition}
\begin{proof}
Let $\phi$ be a pseudo-state for the entropy-energy balance hierarchy at level 1, so that $\phi([H_{G}(h, J), a]) = 0$ for all $a$ of degree 2 and $\phi(a^{*}[H_{G}(h, J), a]) \geq 0$ for all $a$ of degree 1. Consider $a_{k} = (X_{k} + iY_{k})/2$ for some index $k$. Then
\begin{align*}
[H_{G}(h, J), a_{k}] &= -h\sum _{i = 1} ^{n} [Z_{i}, a_{k}] - \frac{J}{2} \sum _{(i, j) \in E(G)} [X_{i}X_{j} + Y_{i}Y_{j}, a_{k}] \\
&= -\frac{h}{2}[Z_{k}, X_{k} + iY_{k}] - \frac{J}{4} \sum _{(i, k) \in E(G)} [X_{i}X_{k} + Y_{i}Y_{k}, X_{k} + iY_{k}] \\
&= -h(iY_{k} + X_{k}) - \frac{J}{2} \sum _{(i, k) \in E(G)} -2(X_{i} + iY_{i})Z_{k}
\shortintertext{so}
a^{*}[H_{G}(h, J), a_{k}] &= -\frac{h}{2}(X_{k} - iY_{k})(X_{k} + iY_{k}) + \frac{J}{2} \sum _{(i, k) \in E(G)} (X_{k} - iY_{k})(X_{i} + iY_{i})Z_{k} \\
&= -h + hZ_{k} + \frac{J}{2} \sum _{(i, k) \in E(G)} (X_{k} - iY_{k})(X_{i} + iY_{i}).
\end{align*}
Summing over all sites and applying $\phi$ then gives
\begin{align*}
0 &\leq \sum _{k} \phi(a_{k}^{*}[H_{G}(h, J), a_{k}]) \\
&= \sum_{k} \phi\left(-h + hZ_{k} + \frac{J}{2} \sum _{(i, k) \in E(G)} (X_{i}X_{k} + Y_{i}Y_{k}) \right) \\
&= -hn - \phi(H_{G}(h, J))
\end{align*}
where in the second line we have used the fact that each edge $(i, j)$ is summed once for $i$ and once for $j$.
Thus $\phi(H_{G}(h, J)) \leq -hn$ for any level-1 pseudostate $\phi$.
\end{proof}

Define
\[ B_{ij} \coloneqq 2I - \frac{Z_{i} + Z_{j}}{2} - \frac{X_{i}X_{j} + Y_{i}Y_{j}}{2}. \]
The following identity will be used in the proof of the lower bound on the spectral gap.

\begin{lemma}
We have
\[ H_{G}(h, J) = -2hn  + (2h - Jd)\left(n - \frac{1}{2} \sum _{i \in V(G)} Z_{i}\right) + J\sum _{(i, j) \in E(G)} B_{ij} \]
whenever $G$ is a $d$-regular graph.
\end{lemma}
\begin{proof}
Expanding gives that the right-hand side is equal to
\begin{align*}
&\phantom{{}={}} (-2hn + 2hn - 2Jdn) - \frac{2h - Jd}{2}\sum _{i \in V(G)} Z_{i} + J\sum _{(i, j) \in E(G)} \left(2 - \frac{Z_{i} + Z_{j}}{2} - \frac{X_{i}X_{j} + Y_{i}Y_{j}}{2}\right) \\
&= -2Jdn - \frac{2h - Jd}{2}\sum _{i \in V(G)} Z_{i} + J\sum _{i \in V(G)} \sum _{(i, j) \in E(G)} \left(2 - \frac{Z_{i} + Z_{j}}{2} - \frac{X_{i}X_{j} + Y_{i}Y_{j}}{2}\right) \\
&= -2Jdn - \frac{2h - Jd}{2}\sum _{i \in V(G)} Z_{i} + 2Jdn - \frac{Jd}{2} \sum _{i \in V(G)} Z_{i} - \frac{J}{2}\sum _{(i, j) \in E(G)} (X_{i}X_{j} + Y_{i}Y_{j}) \\
&= -h\sum _{i \in V(G)} Z_{i} - \frac{J}{2}\sum _{(i, j) \in E(G)} (X_{i}X_{j} + Y_{i}Y_{j}) \\
&= H_{G}(h, J)
\end{align*}
as desired, where we have used the fact that $G$ is $d$-regular in the third line.
\end{proof}

To write the final certificate, we will generally use $\hat{A}$ to denote the noncommutative polynomial in the $S$ and $T$ variables which corresponds to the action of the appropriate operator $A$ on $\wedge^{2}(\mathbb{C}^{2^{n}})$. Concretely, we will have
\[ \hat{B}_{ij} = 2 - \frac{S_{i}^{z} + S_{j}^{z}}{2} - \frac{T_{ij}^{xx} + T_{ij}^{yy}}{2} \]
and
\[ \hat{H}_{G}(h, J) = -2hn  + (2h - Jd)\left(n - \sum _{i \in V(G)} \frac{\hat{Z}_{i}}{2}\right) + J\sum _{(i, j) \in E(G)} \hat{B}_{ij}. \]
Rearranging then gives
\[ \hat{H}_{G}(h, J) - (-hn) - (-hn + (2h - Jd)) = (2h - Jd)\left((n - 1) - \sum _{i \in V(G)} \frac{\hat{Z}_{i}}{2}\right) + J\sum _{(i, j) \in E(G)} \hat{B}_{ij}. \]
From the previous section, we know that the first term on the right-hand side is a sum of squares and that the hierarchy certifies this at level 2. The second term is also a sum of squares. To see this, note that explicitly diagonalizing the $B_{ij}$ only on the sites $i$ and $j$ as a 4-by-4 matrix shows that its eigenvalues are $1$ and $3$, each with multiplicity 2, and so $B_{ij}$ has eigenvalues $1$ and $3$ each with multiplicity $2^{n - 1}$. Thus the action of $B_{ij}$ on $\wedge^{2}(\mathbb{C}^{2^{n}})$ is positive, and by Theorem~\ref{thm:res-reps} the hierarchy certifies that $\hat{B}_{ij}$ is positive at level 2. Thus the entire expression is a sum of squares, certified by the hierarchy at level 2, and combined with the upper bound of $-hn$ on the ground energy certifies a spectral gap of $2h - Jd$.

\subsubsection{Discussion}
As mentioned before, our lower bound of $2h - Jd$ recovers the known bound of $2(h - J)$ for the transverse-field $XY$ model on a 1D chain, which is typically proven by mapping the model to one of free fermions. The Jordan-Wigner transformation used to construct the mapping is nonlocal and of high degree. For standard SDP hierarchies which give lower bounds on the ground energy, free fermion solutions thus do not produce low-degree sum-of-squares certificates. This makes it somewhat more surprising that our hierarchy is able to certify nearly the exact spectral gap of the transverse-field $XY$ model in 1D, and suggests that the representation-theoretic constraints in our noncommutative polynomial system have some nonlocal character which allows them to bound the spectral gap.

\section{Lower bounds on other hierarchies}\label{sec:lowerbounds}
The 1-local Hamiltonian $H_{n} \sum _{i = 1} ^{n} Z_{i}$ is frustration-free, so we also consider some known hierarchies for bounding the spectral gap of frustration-free Hamiltonians. We also consider a hierarchy which does not require frustration-freeness, but which we show is relatively weak. In both cases the hierarchy considered cannot certify a lower bound on the spectral gap on $H_{n}$ except at degree growing with $n$.

\subsection{Squaring the Hamiltonian}
Hierarchies for frustration-free Hamiltonians $H$ are typically based on considering a degree-2 polynomial in $H$ and certifying its positivity. Concretely, if $(H - a)(H - b)$ is a positive operator for $a, b \in \mathbb{R}$ then $H$ cannot have an eigenvalue in the interval $[a, b]$. When taking $b = \lambda_{1}(H)$ to be the ground energy of $H$, a certificate that $(H - a)(H - b)$ is positive then certifies that the spectral gap of $H$ must be at least $a - b$. One can then consider certificates which write $(H - a)(H - b)$ as a sum of squares of low-degree polynomials in the single-site Pauli operators. This approach requires knowing that $b$ is exactly the ground energy of $H$, as if $b$ is only an upper bound then there can still be many excited states in the interval $(\lambda_{1}(H), b])$, and thus is typically restricted to frustration-free Hamiltonians.

In the case of $H_{n} = -\sum _{i = 1} ^{n} Z_{i}$, one can consider it as a polynomial in just the $Z_{i}$, which all commute with each other and satisfy the relations $Z_{i}^{2} = 1$. In \cite[Theorem~1.1]{lpwy2016}, it was shown that $(H_{n} - (-n))(H_{n} - (-n + 2))$ is not $\epsilon$-close in the $L^{\infty}$ norm (considered between functions $\{\pm 1\}^{n} \to \mathbb{R}$ over the Boolean hypercube) to a sum of squares polynomial unless its degree is at least $\Omega(\sqrt{n})$. We show here that the lower bound extends to the case where the $X_{i}$ and $Y_{i}$ are also included, implying a $\sqrt{n}$ lower bound on hierarchies for bounding the spectral gap of this polynomial based on squaring the Hamiltonian.

\begin{proposition}
Let $H$ be a Hamiltonian on $n$ qubits which is a polynomial in $Z_{1}, \dots, Z_{n}$. If $H$ is a sum of squares, then it is a sum of squares of terms which are also polynomials in $Z_{1}, \dots, Z_{n}$.
\end{proposition}
\begin{proof}
First consider the case when $H = A^{*}A$ is itself a square. Let
\[ A = \sum _{\alpha, \beta \in \{0, 1\}^{n}} X^{\alpha}Y^{\beta}P_{\alpha,\beta} \]
be the decomposition of $A$ into a linear combination of Pauli strings where each $P_{\alpha,\beta}$ is a polynomial in the $Z_{i}$, and in every term we never have $\alpha_{i} = \beta_{i} = 1$ for any index $i$. Let
\begin{align*}
D_{i}(z) &= I^{\otimes (i - 1)} \otimes \begin{bmatrix}
1 & 0 \\
0 & z
\end{bmatrix} \otimes I^{\otimes (n - i)}
\shortintertext{which satisfies}
D_{i}(z)^{*}Z_{i}D_{i}(z) &= Z_{i} & D_{i}(z)^{*} &= D(z^{*})_{i}.
\end{align*}
Then
\begin{align*}
D_{i}(z)^{*}X_{i}D(z) &= \begin{bmatrix}0 & z\\z^{*} & 0\end{bmatrix} & D_{i}(z)^{*}Y_{i}D(z) &= \begin{bmatrix}0 & -iz\\iz^{*} & 0\end{bmatrix}
\end{align*}
so
\[ \int _{S^{1}} D(z)_{i}^{*}X_{i}D(z)\, dz = \int _{S^{1}} D(z)_{i}^{*}Y_{i}D(z)\, dz = 0 \]
where the integral is taken with respect to the usual rotationally invariant probability measure (Haar measure) on $S^{1}$.
When $z \in \mathbb{C}^{n}$ is a vector, let $D(z) = \prod _{i = 1} ^{n} D_{i}(z_{i})$. Then since $H$ is a polynomial in the Pauli $Z$ operators, we have
\begin{align*}
H &= \int _{z \in (S^{1})^{n}} D(z)^{*} H D(z)\,dz \\
&= \int _{z \in (S^{1})^{n}} \sum _{\alpha', \beta' \in \{0, 1\}^{n}} \sum _{\alpha, \beta \in \{0, 1\}^{n}} D(z)^{*} P_{\alpha',\beta'}^{*}Y^{\beta'}X^{\alpha'} X^{\alpha}Y^{\beta}P_{\alpha,\beta} D(z)\,dz \\
&= \sum _{\alpha', \beta' \in \{0, 1\}^{n}} \sum _{\alpha, \beta \in \{0, 1\}^{n}} P_{\alpha',\beta'}\prod_{i = 1} ^{n} \left(\int _{S^{1}} D(z)_{i}^{*} Y^{\beta'}X^{\alpha'} X^{\alpha}Y^{\beta}D(z)_{i} \,dz \right)P_{\alpha,\beta} \\
&= \sum _{\alpha, \beta \in \{0, 1\}^{n}} P_{\alpha,\beta}^{*}P_{\alpha,\beta}.
\end{align*}
When $H$ is a sum of more than one square, the same argument applies but with $A^{*}A$ replaced by a sum of terms $\sum _{i} A_{i}^{*} A_{i}$.
\end{proof}

\subsection{Projection onto the low-energy eigenspace} \label{sec:proj}
In this section, we define a hierarchy (suggested by an anonymous reviewer) for bounding the gap of qubit Hamiltonians and then show that it cannot certify a nontrivial bound for any Hamiltonian except at level $\Omega(n)$. We first recall the standard hierarchy for giving a lower bound on the ground state energy of a qubit Hamiltonian $H$, which is
\begin{align}
SDP_{k}(H) := 
\begin{aligned}[t]
\min \quad & \tilde{E}_{2k}[H] \\
\text{subject to} \quad & \tilde{E}_{2k} : \mathcal{A}_{2k} \to \mathbb{C} \\
& \tilde{E}_{2k}[p^{*}p] \geq 0, \quad p \in \mathcal{A}_{k} \\
& \tilde{E}_{2k}[1] = 1 \\
\end{aligned}
\end{align}
where $\mathcal{A}_{2k}$ is the set of all operators acting on $\mathbb{C}^{2^{n}}$ which are polynomials of degree at most $2k$ in the single-site Pauli operators. 

For a Hamiltonian $H$, the optimization problem
\begin{align*}
\min \quad & \operatorname{Tr}(\rho H) \\
\text{subject to} \quad & 1 \succeq \rho \succeq 0 \\
& \operatorname{Tr}(\rho) = 2
\end{align*}
has as its solution the projector onto the subspace spanned by the smallest two eigenvectors, whose value is $\lambda_{0}(H) + \lambda_{1}(H)$. This can be relaxed into a hierarchy expressed for pseudo-expectations $\tilde{E}_{2k}$ defined on the subspaces $\mathcal{A}_{2k}$  as
\begin{equation}
\label{eq:sosgap}
SDPGAP_k(H) := 
\begin{aligned}[t]
\min \quad & \tilde{E}_{2k}[H] \\
\text{subject to} \quad & \tilde{E}_{2k} : \mathcal{A}_{2k} \to \mathbb{C} \\
& \tilde{E}_{2k}[p^{*}p] \geq 0, \quad p \in \mathcal{A}_{k} \\
& \tilde{E}_{2k}[1] = 2 \\
& \tilde{E}_{2k}[p^* p] \leq \operatorname{Tr}(p^* p), \quad p \in \mathcal{A}_{k}
\end{aligned}
\end{equation}
It turns out that $SDPGAP_{k}(H)$ contains no more information than $SDP_{k}(H)$ except at high levels.

\begin{proposition}
Let $H$ be a qubit Hamiltonian with ground energy $E_{0}$. Suppose that 
\begin{equation} \label{eq:cond1}
\frac{1}{2^{n - 1}} \sum_{j=0}^{k} 3^{j} \binom{n}{j} \leq 1.
\end{equation}
Then $SDPGAP_{k}(H) = 2 SDP_{k}(H) \leq 2 E_{0}$.
\end{proposition}

\begin{proof}
It is clear that $SDPGAP_{k}(H) \geq 2 SDP_{k}(H)$, since a feasible solution of $SDPGAP_{k}(H)$ gives a feasible solution of $SDP_{k}(H)$ after scaling by $1/2$. It then suffices to show $2SDP_{k}(H) \geq SDPGAP_{k}(H)$.
Let $M$ be a feasible moment matrix for $SDP_{k}$, that is $M_{ij} = \tilde{E}[P_i^* P_j]$ where $\{P_i\}$ are Pauli operators of degree at most $k$. We claim that $2M$ is a feasible solution for $SDPGAP_{k}$ if $\frac{1}{2^{n - 1}} \sum_{j=0}^k 3^j \binom{n}{j} \leq 1$. Since
\[ \operatorname{Tr}(p^{*}q) = \begin{cases} 2^{n} & p = q \\ 0 & \text{otherwise} \end{cases} \]
where $p$ and $q$ are any two products of Pauli operators, it suffices to check that $2M \preceq 2^{n} I$. We have
\[ 2M \preceq 2\operatorname{Tr}(M) I \preceq 2I\sum_{j=0}^k 3^j \binom{n}{j} \preceq 2^{n} I \] where the first inequality follows because $M \succeq 0$ and the second since $\dim \mathcal{A}_{k} = \sum_{j=0}^k 3^j \binom{n}{j}$.
\end{proof}

For the final lower bound, it suffices to prove

\begin{proposition}
There is a constant $c \in [0, 1]$ and $N \in \mathbb{N}$ such that $\frac{1}{2^{n - 1}} \sum_{j=0}^{k} 3^{j} \binom{n}{j} \leq 1$ whenever $k \leq cn$ and $n \geq N$.
\end{proposition}
\begin{proof}
For any $a > 0$ we have
\begin{align*}
\frac{1}{2^{n - 1}} \sum _{j = 0} ^{k} 3^{j} \binom{n}{j} &= \frac{1}{2^{n - 1}} \sum _{j = 0} ^{k} (3a)^{j} a^{-j} \binom{n}{j} \\
&\leq \frac{1}{2^{n - 1}} \sum _{j = 0} ^{k} (3a)^{k} a^{-j} \binom{n}{j} \\
&\leq (3a)^{k} \frac{1}{2^{n - 1}} \sum _{j = 0} ^{n} a^{-j} \binom{n}{j} \\
&\leq (3a)^{k} 2^{-n + 1} (1 + a)^{n} \\
&= 2\cdot (3a)^{k} \left(\frac{1 + a}{2}\right)^{n} \\
\end{align*}
If $k \leq cn$ then this is bounded by $2\cdot ((3a)^{c}(1 + a)/2)^{n}$. If we take some arbitrary choice of $a > 0$ and then choose $c < \log((1 + a)/2) / \log(3a)$, the desired bound holds whenever $n$ is at least some fixed constant.
\end{proof}

\section{Acknowledgements}
The author thanks Omar Fawzi for suggesting to consider the reference \cite{ffs2024}, and an anonymous reviewer for suggesting to consider the hierarchy in Section~\ref{sec:proj}. The author also thanks Hamza Fawzi and Marius Lemm for helpful discussions which led to the lower bound proven in Section~\ref{sec:proj}. The idea to consider the transverse-field $XY$ model on bounded-degree graphs arose from a conversation between the author and GPT-5.6 Luna.

\nocite{*}
\printbibliography

\end{document}